\documentclass[journal]{IEEEtran}
\usepackage{graphicx}
\usepackage{bbm}
\usepackage{epsfig}
\usepackage{graphics,graphicx,amssymb,amsmath,verbatim}
\usepackage{amssymb}
\usepackage{mathrsfs}
\usepackage{amsfonts}
\usepackage{graphicx}
\usepackage{bbm}
\usepackage{epsfig}
\usepackage{graphics,graphicx,amssymb,amsmath,verbatim}
\usepackage{amssymb}
\usepackage{mathrsfs}
\usepackage{amsfonts}
\usepackage{color}
\usepackage{amsmath}
\usepackage{eurosym}
\usepackage{bbm}
\usepackage{epstopdf}
\usepackage{pstricks}

\newtheorem{theorem}{Theorem}

\newtheorem{corollary}{Corollary}

\newtheorem{lemma}{Lemma}

\newtheorem{proposition}{Proposition}

\newenvironment{proof}[1][Proof]{\noindent\textbf{#1.} }{\ \rule{0.5em}{0.5em}}
\ifCLASSINFOpdf
\else
\fi
\begin{document}

\title{Metamorphic Moving Horizon Estimation}
\author{He~Kong~and~Salah Sukkarieh \thanks{%
This document is an extended version of a technical communique accepted and
to appear in \emph{Automatica}. This work was supported in part by the
Australian Centre for Field Robotics and in part by the Faculty of
Engineering and Information Technologies, University of Sydney, Australia.
The authors are with Australian Centre for Field Robotics, The University of
Sydney, NSW, 2006, Australia. Email: h.kong@acfr.usyd.edu.au;
salah.sukkarieh@sydney.edu.au. Corresponding author: He Kong, Tel.: +61 2
9114 0626.}}
\maketitle

\begin{abstract}
This paper considers a practical scenario where a classical estimation
method might have already been implemented on a certain platform when one
tries to apply more advanced techniques such as moving horizon estimation
(MHE). We are interested to utilize MHE to upgrade, rather than completely
discard, the existing estimation technique. This immediately raises the
question how one can improve the estimation performance gradually based on
the pre-estimator. To this end, we propose a general methodology which
incorporates the pre-estimator with a tuning parameter $\lambda \in \lbrack
0,1]$ into the quadratic cost functions that are usually adopted in MHE. We
examine the above idea in two standard MHE frameworks that have been
proposed in the existing literature. For both frameworks, when $\lambda =0$,
the proposed strategy exactly matches the existing classical estimator; when
the value of $\lambda $ is increased, the proposed strategy exhibits a more
aggressive normalized forgetting effect towards the old data, thereby
increasing the estimation performance gradually.
\end{abstract}




\markboth{}{Shell \MakeLowercase{\textit{et al.}}: A Divide and Conquer Approach to Cooperative Distributed Model
Predictive Control}




\IEEEpeerreviewmaketitle

\section{Introduction}

MHE is a systematic framework to handle constraints in estimation \cite%
{Sukkarieh2001}-\cite{Kerrigan2016}. By far, various forms of MHE have been
proposed. For example, in \cite{Rawlings2001}-\cite{Kong2018}, the cost is
optimized over the initial state and the process noise sequence to minimize
the data fitting error. Other frameworks estimate only the initial state
\cite{Alessandri2003}-\cite{Johansen2014}. The concept of limited memory
filtering has also been adopted in finite impulse response (FIR) filtering
\cite{Zhao2017}-\cite{Shi2016}. MHE and FIR filters are similar in that both
methods only use recent measurements within a time window. However, there
are major differences between them. For example, while the information
contained in the measurements outside the moving horizon is captured in the
so-called arrival cost in MHE \cite[pp. 32-40]{Mayne2009}, such information
is ignored in FIR filtering. A situation that one often encounters when
trying to apply MHE is that some traditional estimators might have already
been implemented. For example, there would be some forms of Kalman filters
embedded in today's most GPS devices. Replacing the existing estimation
methods and related software and hardware by MHE is often time consuming and
costly, if possible. A similar situation is faced by control engineers and
this has motivated works to combine the merits of predictive and linear
methods \cite{Kong2012}-\cite{Maciejowski2013}. Especially, in \cite%
{Kong2013}, a general framework has been proposed to gradually improve
performance using predictive control, incorporating an existing linear
controller.

The question we consider in this paper is to propose a MHE framework to
gradually improve the estimation performance based on a pre-estimator. As
such, we borrow the concept that is originally proposed in \cite{Kong2013}
for the control case, and consolidate the idea in two MHE frameworks that
have been proposed in the existing literature \cite{Rawlings2001}, \cite%
{Alessandri2003}-\cite{Johansen2010}. For both frameworks, we propose a
methodology that can gradually improve the estimation performance with MHE,
incorporating an existing estimator. This is achieved by the introduction of
cost functions parameterized by $\lambda \in \left[ 0,1\right] $. When $%
\lambda $ changes, optimizing the cost functions renders a new estimator, we
thus term the framework metamorphic\footnote{%
As noted in \cite{Kong2013}, metamorphism is the recrystallization of
pre-existing rocks due to physical/chemical changes.}MHE (MMHE). An
advantage of the proposed technique is that it can upgrade an existing
classical method using MHE, thereby obtaining the constraint handling
capabilities of MHE and avoiding the trouble involved in a completely new
design of the estimator. A disadvantage of the framework, compared to
classical estimation techniques, is that one has to solve an optimization
problem at each sampling instant.

\textbf{Notation:} $[a_{1},\cdots ,a_{n}]$ denotes $[a_{1}^{\mathrm{T}%
}\cdots a_{n}^{\mathrm{T}}]^{\mathrm{T}},$ where $a_{1},\cdots ,a_{n}\ $are
scalars, vectors or matrices of proper dimensions; $\mathscr{I}_{i}^{j}$
denotes the set of integers between $i$ and $j$; a set $\mathscr{U}$ $%
\subset $ $\mathbf{R}^{n}$ is a $\mathscr{C}$-set if it is a compact, convex
set that contains the origin in its (non-empty) interior; $diag(M_{1},\cdots
,M_{s})$ denotes a block diagonal matrix with $M_{1},\cdots ,M_{s}$ as its
block diagonal entries, and $diag_{N}(\cdot )$ denotes a block diagonal
matrix with $N$ blocks. $\mathbf{1}_{n}$ denotes a $n$-dimensional column
vector with all its elements as 1.

\section{\label{copy}Metamorphic MHE}

The MHE framework in \cite{Rawlings2001} considers the system%
\begin{equation}
x_{k+1}=Ax_{k}+Gw_{k},\text{ }y_{k}=Cx_{k}+\nu _{k}  \label{plant}
\end{equation}%
where, $x_{k}\in \mathcal{X}\subset \mathbf{R}^{n}$, $w_{k}\in \mathcal{W}%
\subset \mathbf{R}^{m}$ and $\nu _{k}\in \mathcal{V}\subset \mathbf{R}^{p},$
respectively; the pair $(A,C)$ is assumed to be observable; the set $%
\mathcal{X}$ is compact and convex; $\mathcal{W}$ and $\mathcal{V}$ are both
$\mathscr{C}$-sets. The variables $(x_{k},w_{k},y_{k},\nu _{k})$ in (\ref%
{plant}) represent the parameters of the \textit{real} process. In an
optimization-based estimation problem, they have corresponding decision
variables and optimal decision variables, which we denote as $(\chi
_{k},\omega _{k},\eta _{k},\upsilon _{k})$ and $(\widehat{x}_{k},\widehat{w}%
_{k},\widehat{y}_{k},\widehat{\nu }_{k}),$ respectively. The MHE is a
quadratic program (QP) in the form of%
\begin{equation}
\mathcal{M}_{T}:\left\{
\begin{array}{l}
\min\limits_{\chi _{T-N},\mathbf{\omega }_{T-N}^{T-1}}\widetilde{\phi }_{T}%
\text{ s.t.~}\chi _{k}\in \mathcal{X}\text{, }k\in \mathscr{I}_{T-N}^{T} \\
\omega _{k}\in \mathcal{W}\text{, }\upsilon _{k}\in \mathcal{V},\text{ }k\in %
\mathscr{I}_{T-N}^{T-1}%
\end{array}%
\right. ,  \label{MHE}
\end{equation}%
with%
\begin{equation*}
\widetilde{\phi }_{T}=\Theta _{T-N}(\chi _{T-N})+\sum\limits_{k=T-N}^{T-1}%
\left[ \upsilon _{k}^{\mathrm{T}}R^{-1}\upsilon _{k}+\omega _{k}^{\mathrm{T}%
}Q^{-1}\omega _{k}\right] ,
\end{equation*}%
where, $\Theta _{T-N}(\chi _{T-N})=(\chi _{T-N}-\widehat{x}_{T-N}^{m})\Pi
_{T-N}^{-1}(\chi _{T-N}-\widehat{x}_{T-N}^{m})+\widetilde{\phi }_{T-N}^{\ast
},$ $\mathbf{\omega }_{T-N}^{T-1}=\left\{ \omega _{i}\right\} _{i=T-N}^{T-1}$%
, $\chi _{k}=\chi (k-(T-N);\chi _{T-N},\mathbf{\omega }_{T-N}^{k-1}),$ $%
\upsilon _{k}=y_{k}-C\chi _{k};$ the matrix $\Pi _{T-N}$ is the solution to
the ARE%
\begin{equation}
\Pi _{t}=GQG^{\mathrm{T}}+A\Pi _{t-1}A^{\mathrm{T}}-A\Pi _{t-1}R_{t}\Pi
_{t-1}A^{\mathrm{T}},  \label{ARE}
\end{equation}%
with%
\begin{equation*}
R_{t}=C^{\mathrm{T}}(R+C\Pi _{t-1}C^{\mathrm{T}})^{-1}C,
\end{equation*}%
subject to the initial condition $\Pi _{0}$; $\widehat{x}_{T-N}^{m}$ is the
optimal rolling horizon state prediction at time $T-N$, i.e., $\widehat{x}%
_{T-N}^{m}=\widehat{x}_{T-N\mid T-N-1}^{m},$ and $\widetilde{\phi }%
_{T-N}^{\ast }$ is the optimal cost of (\ref{MHE}) at time $T-N$.

\subsection{Embellishing a pre-estimator into MHE}

Assume that for (\ref{plant}), we have the following Luenberger observer or
stationary Kalman filter%
\begin{equation}
\widetilde{x}_{k+1}=A\widetilde{x}_{k}+L(y_{k}-\widetilde{y}_{k}),\text{ }%
\widetilde{y}_{k}=C\widetilde{x}_{k},  \label{observer}
\end{equation}%
where, $L$ is chosen such that $A_{L}=A-LC$ is Schur stable. Define $%
e_{k+1}=x_{k+1}-\widetilde{x}_{k+1}$. Then it holds that%
\begin{equation}
e_{k+1}=A_{L}e_{k}+\vartheta _{k},  \label{errordynamicsMHE}
\end{equation}%
with%
\begin{equation*}
\vartheta _{k}=Gw_{k}-L\nu _{k}\in \mathcal{Q}=G\mathcal{W}\ominus L\mathcal{%
V}.
\end{equation*}%
Note that $\mathcal{Q}$ is also a $\mathscr{C}$-set since both $\mathcal{W}$
and $\mathcal{V}$ are $\mathscr{C}$-sets. Given $\rho (A_{L})<1$, there
exists a robust positively invariant $\mathscr{C}$-set $\mathscr{E}$
satisfying $A_{L}\mathscr{E}\oplus $ $\mathcal{Q}\subseteq $ $\mathscr{E}$
for system (\ref{errordynamicsMHE}) (see \cite{Mayne2009}, pp. 377). Define $%
x_{k}^{e}=\left[ \widetilde{x}_{k},e_{k}\right] $ and $\overline{w}_{k}=%
\left[ w_{k},v_{k}\right] $. From (\ref{plant}), (\ref{observer}), and (\ref%
{errordynamicsMHE}), we have the augmented system%
\begin{equation}
x_{k+1}^{e}=A_{e}x_{k}^{e}+G_{e}\overline{w}_{k},\text{ }%
y_{k}=C_{e}x_{k}^{e}+\nu _{k},  \label{augmented}
\end{equation}%
where,%
\begin{equation*}
\begin{array}{l}
A_{e}=\left[
\begin{array}{cc}
A & LC \\
0 & A_{L}%
\end{array}%
\right] ,G_{e}=\left[
\begin{array}{cc}
0 & L \\
G & -L%
\end{array}%
\right] , \\
C_{e}=\left[
\begin{array}{cc}
C & C%
\end{array}%
\right] .%
\end{array}%
\end{equation*}%
For (\ref{augmented}), we have $x_{k}^{e}\in \overline{\mathcal{X}},$ $%
\overline{w}_{k}\in \overline{\mathcal{W}},$ where $\overline{\mathcal{X}}=%
\mathcal{X\times }\mathscr{E},$ $\overline{\mathcal{W}}=\mathcal{W\times V}.$
The variables $(x_{k}^{e},\overline{w}_{k},y_{k},\nu _{k})$ in (\ref%
{augmented}) represent the parameters of the \textit{real} augmented
process, and we denote $(\chi _{k}^{e},\overline{\omega }_{k},\overline{\eta
}_{k},\upsilon _{k})$ and $(\widehat{x}_{k}^{e},\widehat{\overline{w}}_{k},%
\widehat{y}_{k},\widehat{\nu }_{k})$ as the corresponding decision variables
and the optimal solutions in the optimization, respectively. For notational
ease, we still use $\widehat{y}_{k}$ and $\widehat{\nu }_{k}$ to denote the
optimal output prediction and fitting error for (\ref{augmented}) as for (%
\ref{plant}). For system (\ref{augmented}), consider the constrained
estimation problem
\begin{equation}
\overline{\mathcal{M}}_{T}:\left\{
\begin{array}{l}
\min\limits_{\chi _{T-N}^{e},\overline{\mathbf{\omega }}_{T-N}^{T-1}}%
\overrightarrow{\phi }_{T}\text{ s.t.~}\chi _{k}^{e}\in \overline{\mathcal{X}%
}\text{, }k\in \mathscr{I}_{T-N}^{T} \\
\text{ \ \ \ }\overline{\omega }_{k}\in \overline{\mathcal{W}}\text{, }%
\upsilon _{k}\in \mathcal{V},\text{ }k\in \mathscr{I}_{T-N}^{T-1}%
\end{array}%
\right. ,  \label{MetamorphicFIE}
\end{equation}%
where, $\chi _{k}^{e}=\chi ^{e}(k;\chi _{T-N}^{e},\overline{\mathbf{\omega }}%
_{T-N}^{k-1}),$ $\upsilon _{k}=y_{k}-C\chi _{k}^{e},$ $\lambda \in \left[ 0,1%
\right] $, $\overline{\mathbf{\omega }}_{T-N}^{T-1}=\left\{ \overline{\omega
}_{i}\right\} _{i=T-N}^{T-1},$%
\begin{equation*}
\begin{array}{l}
\overrightarrow{\phi }_{T}=\lambda (\chi _{T-N}^{e}-\widehat{x}%
_{T-N}^{em})\Phi _{T-N}^{-1}(\chi _{T-N}^{e}-\widehat{x}_{T-N}^{em})+\lambda
\overrightarrow{\phi }_{T-N}^{\ast } \\
+\sum\limits_{k=T-N}^{T-1}\left[ (1-\lambda )\overline{\omega }_{k}^{\mathrm{%
T}}M\overline{\omega }_{k}+\lambda (\upsilon _{k}^{\mathrm{T}}R^{-1}\upsilon
_{k}+\omega _{k}^{\mathrm{T}}Q^{-1}\omega _{k})\right] ,%
\end{array}%
\end{equation*}%
in which, $R,Q,M>0$; $\Phi _{T-N}$ is a positive definite matrix to be
discussed in the sequel; $\overrightarrow{\phi }_{T-N}^{\ast }$ is the
optimal cost of (\ref{MetamorphicFIE}) at time $T-N$, and thus is a constant
parameter and can be safely ignored in the optimization; $\widehat{x}%
_{T-N}^{em}$ is the optimal moving horizon state prediction at time $T-N$,
i.e., $\widehat{x}_{T-N}^{em}=\widehat{x}_{T-N\mid T-N-1}^{em}$.

When $\lambda =0$, one has $\overrightarrow{\phi }_{T}=\sum%
\limits_{k=T-N}^{T-1}\overline{\omega }_{k}^{\mathrm{T}}M\overline{\omega }%
_{k}$. Given $0\in \overline{\mathcal{W}}$, the optimal decision variables
are $\widehat{\overline{w}}_{i}=0,$ for $i\in \mathscr{I}_{T-N}^{T-1}.$ In
this case, the optimal decision variables $(\widehat{x}_{k}^{e},\widehat{%
\overline{w}}_{k},\widehat{y}_{k},\widehat{\nu }_{k})$ satisfy $\widehat{x}%
_{k+1}^{e}=A_{e}\widehat{x}_{k}^{e}$, $\widehat{y}_{k}=C_{e}\widehat{x}%
_{k}^{e}$, i.e., the strategy reduces to a deterministic observer with the
same gain as the pre-estimator (\ref{observer}). When $\lambda =1$, one has%
\begin{equation*}
\begin{array}{l}
\overrightarrow{\phi }_{T}=(\chi _{T-N}^{e}-\widehat{x}_{T-N}^{em})\Phi
_{T-N}^{-1}(\chi _{T-N}^{e}-\widehat{x}_{T-N}^{em}) \\
+\sum\limits_{k=T-N}^{T-1}\left[ \overline{\omega }_{k}^{\mathrm{T}}%
\overline{Q}\overline{\omega }_{k}+\upsilon _{k}^{\mathrm{T}}R^{-1}\upsilon
_{k}\right] ,%
\end{array}%
\end{equation*}%
with $\overline{Q}=diag(Q^{-1},0)\geq 0$. This is not a well-posed case
since positive definiteness is required for the weight on $\overline{\omega }%
_{k}$. Thus, we will only consider the cases of $\lambda \in (0,1)$.
Dividing $\overrightarrow{\phi }_{T}$ by $\lambda $ gives:%
\begin{equation}
\begin{array}{l}
\overline{\phi }_{T}=\lambda ^{-1}\overrightarrow{\phi }_{T}=\sum%
\limits_{k=T-N}^{T-1}\left[ \overline{\omega }_{k}^{\mathrm{T}}Q_{e}^{-1}%
\overline{\omega }_{k}+\upsilon _{k}^{\mathrm{T}}R^{-1}\upsilon _{k}\right]
\\
+(\chi _{T-N}^{e}-\widehat{x}_{T-N}^{em})\Phi _{T-N}^{-1}(\chi _{T-N}^{e}-%
\widehat{x}_{T-N}^{em}),%
\end{array}
\label{normalizedcost}
\end{equation}%
where,%
\begin{equation*}
Q_{e}^{-1}=\frac{1-\lambda }{\lambda }M+diag(Q^{-1},0)>0,
\end{equation*}%
given $\lambda \in (0,1)$ and $M>0$. Moreover, one can consider a
constrained estimation problem replacing $\overrightarrow{\phi }_{T}$ in (%
\ref{MetamorphicFIE}) with $\overline{\phi }_{T}$ (\ref{normalizedcost}).
Doing so will not affect optimality or stability.

\subsection{Stability ingredients for metamorphic MHE}

\begin{proposition}
\label{prop11}Assume that $Q,R,M>0$, $(A,C)$ is observable. For $\lambda \in
(0,1)$, we have: (i) $(A_{e},C_{e})$ is observable; (ii) if $(A,GQ^{-1/2})$
is controllable, then $(A_{e},G_{e}Q_{e}^{-1/2})$ is controllable, with $%
Q_{e}$ defined in (\ref{normalizedcost}).
\end{proposition}

\begin{proof}
(i): Note that the poles of $A_{e}-L_{1}C_{e}$ can be arbitrarily placed
within the unit circle by $L_{1}=\left[
\begin{array}{cc}
L & 0%
\end{array}%
\right] ^{\mathrm{T}}$. (ii): Given $(A,GQ^{-1/2})$ is controllable, there
exists a matrix $K$ so that the poles of $A+GQ^{-1/2}K$ can be placed
anywhere in the unit circle. Given $Q_{e}^{-1}>0,$ for $\lambda \in (0,1)$,
there exists a unique $Q_{s}>0$ such that $Q_{s}^{2}=Q_{e}^{-1}$ \cite%
{Laub2005} (pp. 101). Denote%
\begin{equation*}
K_{e}=Q_{s}^{-1}\left[
\begin{array}{cc}
0 & Q^{-1/2}K \\
-C & -C%
\end{array}%
\right] .
\end{equation*}%
It can be verified that the poles of $A_{e}+G_{e}Q_{e}^{-1/2}K_{e}$ can be
arbitrarily placed within the unit circle by $K_{e}$.
\end{proof}

When one replaces $\overrightarrow{\phi }_{T}$ in (\ref{MetamorphicFIE})
with $\overline{\phi }_{T}$ (\ref{normalizedcost}), the associated ARE for
system (\ref{augmented}) is%
\begin{equation}
\Phi _{T}=G_{e}Q_{e}G_{e}^{\mathrm{T}}+A_{e}\Phi _{T-1}A_{e}^{\mathrm{T}%
}-A_{e}R_{e}A_{e}^{\mathrm{T}}  \label{AREaugment}
\end{equation}%
with $\Phi _{0}$ as the initial condition,%
\begin{equation*}
R_{e}=\Phi _{T-1}C_{e}^{\mathrm{T}}(R+C_{e}\Phi _{T-1}C_{e}^{\mathrm{T}%
})^{-1}C_{e}\Phi _{T-1},
\end{equation*}%
and $Q_{e}$ being defined in (\ref{normalizedcost}). Without constraints,
one obtains the metamorphic Kalman filter%
\begin{equation*}
\widehat{x}_{T}^{e}=A_{e}\widehat{x}_{T-1}^{e}+L_{e}(y_{T}-C_{e}A_{e}%
\widehat{x}_{T-1}^{e}),
\end{equation*}%
where,%
\begin{equation*}
L_{e}=A_{e}\Phi _{T-1}C_{e}^{\mathrm{T}}(R+C_{e}\Phi _{T-1}C_{e}^{\mathrm{T}%
})^{-1}.
\end{equation*}%
We have the following results regarding the invertibility of $\Phi _{T}$ (%
\ref{AREaugment}) for later use. The results can be proved by following the
ideas in \cite{Bitmead1985}, and thus the proof is skipped here.

\begin{lemma}
\label{lm0}Assume that $Q,R,M,\Phi _{0}$ are positive definite, $(A,C)$ is
observable. For $\lambda \in (0,1),$ we have $\Phi _{k}>0$, for all $k\geq 0$%
, if either of the following two conditions is satisfied: (i) $(A,GQ^{-1/2})$
is controllable, and $\Phi _{0}\geq \Phi _{\infty };$ (ii) $G$ and $L$ are
both nonsingular.
\end{lemma}

\begin{theorem}
\label{th2}Assume that $\Phi _{0}$ is chosen independently of $\lambda $,
and $\Phi _{T}$ is updated according to the ARE (\ref{AREaugment}). Suppose
either of the two conditions in Lemma \ref{lm0} is satisfied, i.e., $\Phi
_{k}>0$, for $k\geq 0$, then for $\lambda \in (0,1)$, we have $\frac{d\Phi
_{k}}{d\lambda }\geq 0.$
\end{theorem}

\begin{proof}
We prove the above result by induction. Since $\Phi _{0}$ is independent of $%
\lambda $, for $k=0,$ we have $\frac{d\Phi _{0}}{d\lambda }=0$. When $k=1,$
from (\ref{AREaugment}), we have%
\begin{equation*}
\begin{array}{l}
\Phi _{1}=G_{e}Q_{e}G_{e}^{\mathrm{T}}+A_{e}\Phi _{0}A_{e}^{\mathrm{T}} \\
-A_{e}\Phi _{0}C_{e}^{\mathrm{T}}(R+C_{e}\Phi _{0}C_{e}^{\mathrm{T}%
})^{-1}C_{e}\Phi _{0}A_{e}^{\mathrm{T}}.%
\end{array}%
\end{equation*}%
Note that the second and third items on the right hand side of the above
equality is independent of $\lambda .$ Therefore, it holds that%
\begin{equation}
\frac{d\Phi _{1}}{d\lambda }=G_{e}\frac{dQ_{e}}{d\lambda }G_{e}^{\mathrm{T}},
\label{difOmega}
\end{equation}%
where we have used the rule%
\begin{equation}
\frac{dXY}{dx}=\frac{dX}{dx}Y+X\frac{dY}{dx},  \label{rule1}
\end{equation}%
in which, $X$ and $Y$ are two matrices of compatible dimensions, and $x$ is
a scalar. Note, for $Q_{e}^{-1}$, we have (\ref{QbarDiff}). Hence the
following holds%
\begin{equation}
\frac{dQ_{e}}{d\lambda }=\frac{1}{\lambda ^{2}}Q_{e}MQ_{e}>0.  \label{QEDIFF}
\end{equation}%
By combing the above expression and (\ref{difOmega}), we have%
\begin{equation*}
\frac{d\Phi _{1}}{d\lambda }=G_{e}\frac{dQ_{e}}{d\lambda }G_{e}^{\mathrm{T}}=%
\frac{1}{\lambda ^{2}}G_{e}Q_{e}MQ_{e}G_{e}^{\mathrm{T}}\geq 0.
\end{equation*}%
Now, for $k=2,$ we have%
\begin{equation*}
\begin{array}{l}
\frac{d\Phi _{2}}{d\lambda }=G_{e}\frac{dQ_{e}}{d\lambda }G_{e}^{\mathrm{T}%
}+A_{e}\frac{d\Phi _{1}}{d\lambda }A_{e}^{\mathrm{T}}-A_{e}\frac{d\Phi _{1}}{%
d\lambda }C_{e}^{\mathrm{T}}\overline{R}^{-1}C_{e}\Phi _{1}\overline{A}^{%
\mathrm{T}} \\
-A_{e}\Phi _{1}C_{e}^{\mathrm{T}}\frac{d\overline{R}^{-1}}{d\lambda }%
C_{e}\Phi _{1}A_{e}^{\mathrm{T}}-A_{e}\Phi _{1}C_{e}^{\mathrm{T}}\overline{R}%
^{-1}C_{e}\frac{d\Phi _{1}}{d\lambda }A_{e}^{\mathrm{T}} \\
=G_{e}\frac{dQ_{e}}{d\lambda }G_{e}^{\mathrm{T}}+A_{e}\left( \frac{d\Phi _{1}%
}{d\lambda }-\frac{d\Phi _{1}}{d\lambda }C_{e}^{\mathrm{T}}\overline{R}%
^{-1}C_{e}\Phi _{1}\right. \\
\left. -\Phi _{1}C_{e}^{\mathrm{T}}\frac{d\overline{R}^{-1}}{d\lambda }%
C_{e}\Phi _{1}\overline{A}^{\mathrm{T}}-\Phi _{1}C_{e}^{\mathrm{T}}\overline{%
R}^{-1}C_{e}\frac{d\Phi _{1}}{d\lambda }\right) A_{e}^{\mathrm{T}},%
\end{array}%
\end{equation*}%
where, $\overline{R}=R+C_{e}\Phi _{1}C_{e}^{\mathrm{T}}$. Since $\frac{d\Phi
_{1}}{d\lambda }\geq 0$, there exists a unique $\Delta \geq 0$ such that $%
\Delta ^{2}=\frac{d\Phi _{1}}{d\lambda }$ \cite{Laub2005} (pp. 101). Based
on the expression of $\frac{d\Phi _{1}}{d\lambda }$, (\ref{rule1}), and the
fact that%
\begin{equation}
\frac{dX^{-1}}{dx}=-X^{-1}\frac{dX}{dx}X^{-1},  \label{rule2}
\end{equation}%
in which, $X$ is a nonsingular matrix with its components as functions of a
scalar $x$, one has%
\begin{equation*}
\begin{array}{c}
\frac{d\overline{R}^{-1}}{d\lambda }=-\overline{R}^{-1}\frac{d(R+C_{e}\Phi
_{1}C_{e}^{\mathrm{T}})}{d\lambda }\overline{R}^{-1} \\
=-\overline{R}^{-1}C_{e}\Delta ^{2}C_{e}^{\mathrm{T}}\overline{R}^{-1}\leq 0.%
\end{array}%
\end{equation*}%
Therefore, we have%
\begin{equation*}
\begin{array}{l}
A_{e}\left( \frac{d\Phi _{1}}{d\lambda }-\frac{d\Phi _{1}}{d\lambda }C_{e}^{%
\mathrm{T}}\overline{R}^{-1}C_{e}\Phi _{1}\right. \\
\left. -\Phi _{1}C_{e}^{\mathrm{T}}\frac{d\overline{R}^{-1}}{d\lambda }%
C_{e}\Phi _{1}-\Phi _{1}C_{e}^{\mathrm{T}}\overline{R}^{-1}C_{e}\frac{d\Phi
_{1}}{d\lambda }\right) A_{e}^{\mathrm{T}} \\
=A_{e}\left( \Delta ^{2}-\Delta ^{2}C_{e}^{\mathrm{T}}\overline{R}%
^{-1}C_{e}\Phi _{1}-\Phi _{1}\overline{C}^{\mathrm{T}}\overline{R}^{-1}%
\overline{C}\Delta ^{2}\right. \\
\left. +\Phi _{1}C_{e}^{\mathrm{T}}\overline{R}^{-1}C_{e}\Delta ^{2}C_{e}^{%
\mathrm{T}}\overline{R}^{-1}C_{e}\Phi _{1}\right) A_{e}^{\mathrm{T}} \\
=A_{e}(\Delta -\Phi _{1}C_{e}^{\mathrm{T}}\overline{R}^{-1}C_{e}\Delta
)(\Delta -\Phi _{1}C_{e}^{\mathrm{T}}\overline{R}^{-1}C_{e}\Delta )^{\mathrm{%
T}}A_{e}^{\mathrm{T}}\geq 0.%
\end{array}%
\end{equation*}%
Since $G_{e}\frac{dQ_{e}}{d\lambda }G_{e}^{\mathrm{T}}\geq 0,$ from the
expression of $\frac{d\Phi _{2}}{d\lambda }$ given after (\ref{QEDIFF}), we
have $\frac{d\Phi _{2}}{d\lambda }\geq 0.$ The above procedure can be
carried out for $k\geq 3.$ Thus, for $k\geq 0$ and $\lambda \in (0,1)$, it
holds that$\frac{d\Phi _{k}}{d\lambda }\geq 0$.
\end{proof}

When $\Phi _{0}>\Phi _{\infty },$ the assumption that $\Phi _{0}$ is chosen
independently of $\lambda $ can be satisfied by selecting a sufficiently
large $\Phi _{0}$. Therefore, the results in Theorem \ref{th2} can be
applied for this case. When $\Phi _{0}=\Phi _{\infty },$ $\Phi _{k}=\Phi
_{\infty },$ for all $k\geq 0$, e.g., $\Phi _{0}$ is dependent of $\lambda $%
, as $\Phi _{\infty }$ is. We have the following results complementary to
Theorem \ref{th2}.

\begin{corollary}
\label{corco}Assume $Q,R,M$ are positive definite, $(A,C)$ and $%
(A,GQ^{-1/2}) $ are observable and controllable, respectively. Suppose $G$
and $L$ are nonsingular. If $\Phi _{0}=\Phi _{\infty }$, then for $\lambda
\in (0,1)$, it holds that $\frac{d\Phi _{\infty }}{d\lambda }>0.$
\end{corollary}

\begin{proof}
Note that $\Phi _{\infty }$ satisfies%
\begin{equation*}
\Phi _{\infty }=G_{e}Q_{e}G_{e}^{\mathrm{T}}+A_{e}\Phi _{\infty }A_{e}^{%
\mathrm{T}}-A_{e}\Phi _{\infty }C_{e}^{\mathrm{T}}\overline{R}_{\infty
}^{-1}C_{e}\Phi _{\infty }A_{e}^{\mathrm{T}},
\end{equation*}%
where, $\overline{R}_{\infty }=R+C_{e}\Phi _{\infty }C_{e}^{\mathrm{T}}.$
Denote $\overline{L}=A_{e}\Phi _{\infty }C_{e}^{\mathrm{T}}\overline{R}%
_{\infty }^{-1}$ and $\overline{A}_{L}=A_{e}-\overline{L}C_{e}.$
Differentiating the above ARE w.r.t. $\lambda $ on both sides gives us%
\begin{equation*}
\begin{array}{l}
\frac{d\Phi _{\infty }}{d\lambda }=G_{e}\frac{dQ_{e}}{d\lambda }G_{e}^{%
\mathrm{T}}+A_{e}\frac{d\Phi _{\infty }}{d\lambda }A_{e}^{\mathrm{T}}-A_{e}%
\frac{d\Phi _{\infty }}{d\lambda }C_{e}^{\mathrm{T}}\overline{R}_{\infty
}^{-1}C_{e}\Phi _{\infty }A_{e}^{\mathrm{T}} \\
-\text{ }A_{e}\Phi _{\infty }C_{e}^{\mathrm{T}}\frac{d\overline{R}_{\infty
}^{-1}}{d\lambda }C_{e}\Phi _{\infty }A_{e}^{\mathrm{T}}-A_{e}\Phi _{\infty
}C_{e}^{\mathrm{T}}\overline{R}_{\infty }^{-1}C_{e}\frac{d\Phi _{\infty }}{%
d\lambda }A_{e}^{\mathrm{T}} \\
=G_{e}\frac{dQ_{e}}{d\lambda }G_{e}^{\mathrm{T}}+\overline{A}_{L}\frac{d\Phi
_{\infty }}{d\lambda }\overline{A}_{L}^{\mathrm{T}},%
\end{array}%
\end{equation*}%
where we have used (\ref{rule1}) and (\ref{rule2}). When $G$ and $L$ are
nonsingular, one has that $G_{e}$ is full rank, e.g.,%
\begin{equation}
\overline{A}_{L}\frac{d\Phi _{\infty }}{d\lambda }\overline{A}_{L}^{\mathrm{T%
}}-\frac{d\Phi _{\infty }}{d\lambda }=-G_{e}\frac{dQ_{e}}{d\lambda }G_{e}^{%
\mathrm{T}}<0,  \label{contra}
\end{equation}%
given (\ref{QEDIFF}). If $(A,C)$ and $(A,GQ^{-1/2})$ are observable and
controllable, respectively, $(A_{e},C_{e})$ and $(A_{e},G_{e}Q_{e}^{-1/2})$
are observable and controllable, respectively, as proved in Proposition \ref%
{prop11}. Under such conditions, $\overline{A}_{L}$ is Schur stable (see,
e.g., Theorem 4.1 in \cite{Goodwin1984}). Therefore, the discrete-time
Lyapunov matrix equation (\ref{contra}) has a unique positive definite
solution, i.e., $\frac{d\Phi _{\infty }}{d\lambda }>0$.
\end{proof}

We assume $R,M$ to be independent of $\lambda $. It then can be verified that%
\begin{equation}
\frac{d(Q_{e}^{-1})}{d\lambda }=-\frac{1}{\lambda ^{2}}M<0.  \label{QbarDiff}
\end{equation}%
Differentiating (\ref{normalizedcost}) w.r.t. $\lambda $ gives%
\begin{equation*}
\begin{array}{l}
\frac{d\overline{\phi }_{T}}{d\lambda }=-\frac{1}{\lambda ^{2}}%
\sum\limits_{k=0}^{T-1}\overline{\omega }_{k}^{\mathrm{T}}M\overline{\omega }%
_{k} \\
-(\chi _{T-N}^{e}-\widehat{x}_{T-N}^{em})^{\mathrm{T}}\Phi _{T-N}^{-1}\frac{%
d\Phi _{T-N}}{d\lambda }\Phi _{T-N}^{-1}(\chi _{T-N}^{e}-\widehat{x}%
_{T-N}^{em}),%
\end{array}%
\end{equation*}%
where fact (\ref{rule2}) is used, in which, $X$ is a nonsingular matrix. If $%
\Phi _{0}$ is independent of $\lambda $, $\frac{d\Phi _{0}}{d\lambda }=0$,
from Theorem \ref{th2}, one has$\frac{d\Phi _{k}}{d\lambda }\geq 0$. If $%
\Phi _{0}$ is dependent of $\lambda ,$ e.g., $\Phi _{0}$ is chosen to be the
steady-state solution of the ARE (\ref{AREaugment}), it is established in
Corollary \ref{corco} that $\frac{d\Phi _{\infty }}{d\lambda }\geq 0$. Note
that for the above both cases, the \textit{inverse} of matrix $\Phi _{k}$ or
$\Phi _{\infty }$ is used in MMHE as the weighting on the state estimate
obtained using the old measurement outside the moving horizon. In other
words, an increase of $\lambda $ results in a more aggressive forgetting
effect towards the estimate using old data. Therefore, we remark that the
increase of $\lambda $ stands for the designers' willingness to rely on more
recent, rather than relatively old, data. Based on the above results and by
following the steps in \cite{Rawlings2001}, one can establish the stability
of both metamorphic FIE and MHE in the sense of an observer. Due to limited
space, details of the stability analysis are skipped here.

\section{\label{main2}Metamorphic MHE of the Initial State}

The MHE framework in \cite{Alessandri2003}-\cite{Johansen2010} considers the
system%
\begin{equation}
x_{k+1}=Ax_{k}+Bu_{k}+w_{k},\text{ }y_{k}=Cx_{k}+\nu _{k}  \label{plant2}
\end{equation}%
where, $x_{k}\in \mathbf{R}^{n}$, $w_{k}\in \mathcal{W}\subset \mathbf{R}%
^{n} $ and $\nu _{k}\in \mathcal{V}\subset \mathbf{R}^{p}$ with $\mathcal{W}$
and $\mathcal{V}$ standing for two $\mathscr{C}$-sets; $(A,C)$ is
observable. At each time step, the MHE strategy in \cite{Alessandri2003}
solves%
\begin{equation}
\mathcal{P}_{t}:\left\{
\begin{array}{l}
\widehat{x}_{t-N,t}^{o}\text{\textbf{\ }}\mathbf{=\arg \min }\text{ }J_{t}
\\
\text{s.t.~}\widehat{x}_{i+1,t}=A\widehat{x}_{i,t}+Bu_{i},\text{ }i\in %
\mathscr{I}_{t-N}^{t-1} \\
\text{ \ \ \ }\widehat{y}_{i,t}=C\widehat{x}_{i,t},\text{ }i\in \mathscr{I}%
_{t-N}^{t}%
\end{array}%
\right. ,  \label{pt}
\end{equation}%
where,
\begin{equation*}
J_{t}=\mu \left\Vert \widehat{x}_{t-N,t}-\overline{x}_{t-N,t}\right\Vert
^{2}+\sum\limits_{i=t-N}^{t}\left\Vert y_{i}-C\widehat{x}_{i,t}\right\Vert
^{2},
\end{equation*}%
with $\mu \geqslant 0$, $\overline{x}_{t-N}$ is an \textit{a priori }%
estimation of $x_{t-N}$, $N\geq n$. In \cite{Alessandri2003}, it is assumed
that for $t=N$, one already has $\overline{x}_{0}$ as a priori\textit{\ }%
estimate; for $t=N+1,N+2,\cdots ,$ the priori\textit{\ }estimation $%
\overline{x}_{t-N}$ is updated via $\overline{x}_{t-N,t}=A\widehat{x}%
_{t-N-1,t-1}^{o}+Bu_{t-N-1}$, where $\widehat{x}_{t-N-1,t-1}^{o}$ is the
optimal estimate at the previous estimation step. Motivated by the fact that
$\widehat{x}_{i+1,t}$ is obtained by updating the system dynamics from $%
\widehat{x}_{i,t}$ in open-loop in (\ref{pt}), \cite{Johansen2010} has
proposed to embed a Luenberger observer into (\ref{pt}):%
\begin{equation}
\overline{\mathcal{P}}_{t}:\left\{
\begin{array}{l}
\widehat{x}_{t-N,t}^{o}\text{\textbf{\ }}\mathbf{=\arg \min }\text{ }%
\overline{J}_{t} \\
\text{s.t.~}\widehat{x}_{i+1,t}=A\widehat{x}_{i,t}+Bu_{i} \\
+\text{ }L(y_{i}-\widehat{y}_{i,t}),\text{ }i\in \mathscr{I}_{t-N}^{t-1} \\
\widehat{y}_{i,t}=C\widehat{x}_{i,t},\text{ }i\in \mathscr{I}_{t-N}^{t}%
\end{array}%
\right. ,  \label{ptbar}
\end{equation}%
where,%
\begin{equation*}
\overline{J}_{t}=\mu \left\Vert \widehat{x}_{t-N,t}-\overline{x}%
_{t-N,t}\right\Vert ^{2}+\left\Vert W(\mathbf{y}_{t-N}^{t}-\widehat{\mathbf{y%
}}_{t-N}^{t,t})\right\Vert ^{2},
\end{equation*}%
with $\mu \geqslant 0$, $W\in \mathbf{R}^{n\times (N+1)p}$, and%
\begin{equation}
\mathbf{y}_{t-N}^{t}=[y_{t-N},\cdots ,y_{t}],\text{ }\widehat{\mathbf{y}}%
_{t-N}^{t,t}=[\widehat{y}_{t-N,t},\cdots ,\widehat{y}_{t,t}]  \label{yyhat}
\end{equation}%
and $L$ is chosen such that $A_{L}=A-LC$ is Schur stable. Different from
\cite{Alessandri2003}, for $t=N+1,N+2,\cdots ,$ the priori\textit{\ }%
estimation $\overline{x}_{t-N,t}$ in \cite{Johansen2010} is updated via%
\begin{equation}
\left\{
\begin{array}{l}
\overline{x}_{t-N,t}=A\widehat{x}_{t-N-1,t-1}^{o}+Bu_{t-N-1} \\
+\text{ }L(y_{t-N-1}-\widehat{y}_{t-N-1,t-1}^{o}) \\
\widehat{y}_{t-N-1,t-1}^{o}=C\widehat{x}_{t-N-1,t-1}^{o}%
\end{array}%
\right. ,  \label{estimateupdate}
\end{equation}%
where $\widehat{x}_{t-N-1,t-1}^{o}$ is the optimal estimate at the previous
estimation step and $\overline{x}_{0,N}=\overline{x}_{0}$. \cite%
{Johansen2010} has shown that the convergence of the MHE (\ref{ptbar}) only
depends on the Schur stability of $A_{L}$ and is independent of $\mu $,
given the introduction of the pre-estimation.

\subsection{A parameterized cost for MHE with pre-estimation}

The MHE problem we consider takes the form of (\ref{ptbar}), with $\overline{%
J}_{t}$ being replaced by%
\begin{equation*}
\begin{array}{l}
\overrightarrow{J}_{t}=(1-\lambda )\overline{\mu }\left\Vert \widehat{x}%
_{t-N,t}-\overline{x}_{t-N,t}\right\Vert ^{2} \\
+\text{ }\lambda \left( \sum\limits_{i=t-N}^{t}\left\Vert y_{i}-C\widehat{x}%
_{i,t}\right\Vert ^{2}+\mu \left\Vert \widehat{x}_{t-N,t}-\overline{x}%
_{t-N,t}\right\Vert ^{2}\right) ,%
\end{array}%
\end{equation*}%
with $\overline{\mu },$ $\mu \geqslant 0$, $\lambda \in \lbrack 0,1]$. We
consider the simple case with scaler weightings without constraints. The
results can be extended to general cases with matrix weightings and
constraints \cite{Johansen2014}. As in \cite{Johansen2010}, the priori%
\textit{\ }estimation $\overline{x}_{t-N,t}$ is updated via (\ref%
{estimateupdate}). We firstly rewrite $\overrightarrow{J}_{t}$%
\begin{equation}
\overrightarrow{J}_{t}=\overline{\lambda }\left\Vert \widehat{x}_{t-N,t}-%
\overline{x}_{t-N,t}\right\Vert ^{2}+\lambda \left\Vert \mathbf{y}_{t-N}^{t}-%
\widehat{\mathbf{y}}_{t-N}^{t,t}\right\Vert ^{2},  \label{costcost}
\end{equation}%
where,%
\begin{equation*}
\overline{\lambda }=\lambda \mu +(1-\lambda )\overline{\mu }.
\end{equation*}%
When $\lambda =0,$ $\overrightarrow{J}_{t}(\overline{x}_{t-N,t},\mathcal{I}%
_{t})=\overline{\mu }\left\Vert \widehat{x}_{t-N,t}-\overline{x}%
_{t-N,t}\right\Vert ^{2},$ and the optimal estimate $\widehat{x}_{t-N,t}^{o}=%
\overline{x}_{t-N,t}$, i.e., the proposed strategy matches the existing
Luenberger observer. When $\lambda =1$, the term $(1-\lambda )\left\Vert
\widehat{x}_{t-N,t}-\overline{x}_{t-N,t}\right\Vert ^{2}$ disappears, i.e.,
the proposed method reduces to a similar MHE strategy with that of \cite%
{Johansen2010}. In general, $\overline{\mu }$ and $\mu $ are two nonnegative
scalars to be selected by the user. The special cases of when they are
equal, or when one of them is zero with the other being positive, can be
analyzed easily. In the following, we consider the case of $\lambda \in
(0,1) $ with $\overline{\mu }$ and $\mu $ having different positive values.
For this case, the increase of $\lambda $ stands for the designers' desire
to forget relatively old data and make sure the nominal model prediction
track the new measurements closely. This can be more clearly seen if we
divide both sides of (\ref{costcost}) by $\lambda $, i.e.,%
\begin{equation}
\begin{array}{l}
\widetilde{J}_{t}=\lambda ^{-1}\overrightarrow{J}_{t} \\
=\lambda ^{-1}\overline{\lambda }\left\Vert \widehat{x}_{t-N,t}-\overline{x}%
_{t-N,t}\right\Vert ^{2}+\left\Vert \mathbf{y}_{t-N}^{t}-\widehat{\mathbf{y}}%
_{t-N}^{t,t}\right\Vert ^{2}.%
\end{array}
\label{newcost}
\end{equation}%
In $\widetilde{J}_{t}$, the weightings on $\left\Vert \widehat{x}_{t-N,t}-%
\overline{x}_{t-N,t}\right\Vert ^{2}$ and $\left\Vert \mathbf{y}_{t-N}^{t}-%
\widehat{\mathbf{y}}_{t-N}^{t,t}\right\Vert ^{2}$ is $\lambda ^{-1}\overline{%
\lambda }$ and $1$, respectively.

\begin{proposition}
\label{prop1} Assume $\lambda \in (0,1)$, $\overline{\mu }$ and $\mu $ take
different positive values. Then the following results hold:

(i) if $(1-\lambda )\overline{\mu }=\lambda (1-\mu )$, then $\lambda ^{-1}%
\overline{\lambda }=1$; a necessary condition for $\lambda ^{-1}\overline{%
\lambda }=1$ is that $0<\mu <1;$

(ii) if $(1-\lambda )\overline{\mu }>\lambda (1-\mu ),$ then $\lambda ^{-1}%
\overline{\lambda }>1$; moreover, if $\mu \geq 1,$ one always has $\lambda
^{-1}\overline{\lambda }>1;$

(iii) if $(1-\lambda )\overline{\mu }<\lambda (1-\mu ),$ then $\lambda ^{-1}%
\overline{\lambda }<1$; a necessary condition for $\lambda ^{-1}\overline{%
\lambda }<1$ is that $0<\mu <1;$

(iv) one always has $\frac{d(\lambda ^{-1}\overline{\lambda })}{d\lambda }<$
$0$.
\end{proposition}

\begin{proof}
(i): Assume that $\lambda \in (0,1)$, $\overline{\mu }$ and $\mu $ take
different positive values. $\lambda ^{-1}\overline{\lambda }=1$ is
equivalent to $\lambda \mu +(1-\lambda )\overline{\mu }=\lambda $, which
reduces to $(1-\lambda )\overline{\mu }=\lambda (1-\mu ).$ Since $(1-\lambda
)\overline{\mu }>0$ for $\lambda \in (0,1)$, $\overline{\mu }>0$, one must
have that $\lambda (1-\mu )>0.$ Therefore, it must hold that $0<\mu <1.$
Part (i) is proved. (ii) and (iii): These results can be proved by following
similar arguments with part (i). (iv): It can be derived%
\begin{equation*}
\frac{d(\lambda ^{-1}\overline{\lambda })}{d\lambda }=-\frac{\overline{\mu }%
}{\lambda ^{2}}<0,
\end{equation*}%
for $\lambda \in (0,1)$, $\overline{\mu }>0$. This completes the proof.
\end{proof}

In the following, we consider the MHE problem (\ref{ptbar}) with cost $%
\overrightarrow{J}_{t}$ being replaced by $\widetilde{J}_{t}$ (\ref{newcost}%
). Doing so will not affect the optimal solution or stability, since the
cost is only changed by a positive scalar. Denote%
\begin{equation*}
\mathbf{C}=diag_{N+1}(C,\cdots ,C)
\end{equation*}%
and $\widehat{x}_{t-N,t}^{o}$ as the optimal solution to the MHE problem (%
\ref{ptbar}) with cost $\widetilde{J}_{t}$ (\ref{newcost}), and the
estimation error as%
\begin{equation*}
e_{t-N}=x_{t-N}-\widehat{x}_{t-N,t}^{o}.
\end{equation*}

\subsection{Estimation error analysis}

Denote%
\begin{equation*}
\begin{array}{l}
\widehat{\mathbf{x}}_{t-N}^{t,t}=[\widehat{x}_{t-N,t},\cdots ,\widehat{x}%
_{t,t}],\text{ }\mathbf{x}_{t-N}^{t}=[x_{t-N},\cdots ,x_{t}], \\
\mathbf{w}_{t-N}^{t-1}=[w_{t-N},\cdots ,w_{t-1}],\text{ }\mathbf{v}%
_{t-N}^{t}=[v_{t-N},\cdots ,v_{t}], \\
\mathbf{u}_{t-N}^{t-1}=[u_{t-N},\cdots ,u_{t-1}].%
\end{array}%
\end{equation*}%
From (\ref{plant2}), one can obtain%
\begin{equation}
\mathbf{y}_{t-N}^{t}=\mathbf{\Lambda }x_{t-N}+\mathbf{\Gamma u}_{t-N}^{t-1}+%
\mathbf{\Phi w}_{t-N}^{t-1}+\mathbf{v}_{t-N}^{t},  \label{plantoutput}
\end{equation}%
where, $\mathbf{y}_{t-N}^{t}$, $\mathbf{x}_{t-N}^{t},\mathbf{u}_{t-N}^{t-1},%
\mathbf{w}_{t-N}^{t-1},\mathbf{v}_{t-N}^{t}$ are defined in (\ref{yyhat})
and before (\ref{plantoutput}), respectively, and $\mathbf{\Lambda ,\Gamma
,\Phi }$ can be found in \cite{Johansen2010}. Rewriting the Luenberger
observer's dynamics in (\ref{ptbar}) gives%
\begin{equation*}
\widehat{x}_{i+1,t}=A_{L}\widehat{x}_{i,t}+Bu_{i}+Ly_{i}.
\end{equation*}%
Then, it can be verified that%
\begin{equation}
\widehat{\mathbf{y}}_{t-N}^{t,t}=\mathbf{C}\widehat{\mathbf{x}}_{t-N}^{t,t}=%
\overline{\mathbf{\Lambda }}\widehat{x}_{t-N,t}+\overline{\mathbf{\Gamma }}%
\mathbf{u}_{t-N}^{t-1}+\overline{\mathbf{\Phi }}\mathbf{y}_{t-N}^{t},
\label{ytthat}
\end{equation}%
where, $\widehat{\mathbf{y}}_{t-N}^{t,t}$, $\mathbf{y}_{t-N}^{t}$, $\widehat{%
\mathbf{x}}_{t-N}^{t,t},\mathbf{u}_{t-N}^{t-1}$ are defined in (\ref{yyhat})
and before (\ref{plantoutput}), respectively, and $\overline{\mathbf{\Lambda
}}\mathbf{,\overline{\mathbf{\Gamma }},}\overline{\mathbf{\Phi }}$ can be
found in \cite{Johansen2010}. Denote $\Psi =I_{(N+1)p}-\mathbf{L}_{N}\ $with%
\begin{equation*}
\mathbf{L}_{N}=\left[
\begin{array}{ccccc}
0 & 0 & \cdots & 0 & 0 \\
CL & 0 & \cdots & 0 & 0 \\
CA_{L}L & CL & \cdots & 0 & 0 \\
\vdots & \vdots & \ddots & \vdots & \vdots \\
CA_{L}^{N-1}L & CA_{L}^{N-2}L & \cdots & CL & 0%
\end{array}%
\right] .
\end{equation*}%
We then have the following results on the error dynamics.

\begin{proposition}
\label{propoopyopy} Assume that $\lambda \in (0,1)$, $\overline{\mu }$ and $%
\mu $ take different positive values. Then the estimation error dynamics
takes the following form%
\begin{equation*}
e_{t-N}=\overline{A}_{L}e_{t-N-1}+S^{-1}S_{1}\mathbf{w}%
_{t-N-1}^{t-1}+S^{-1}S_{2}\mathbf{v}_{t-N-1}^{t},
\end{equation*}%
where,%
\begin{equation*}
\begin{array}{l}
\overline{A}_{L}=\lambda ^{-1}\overline{\lambda }\overline{S}^{-1}A_{L},%
\text{ }\overline{S}=\lambda ^{-1}\overline{\lambda }I+\overline{\mathbf{%
\Lambda }}^{\mathrm{T}}\overline{\mathbf{\Lambda }} \\
S_{1}=\left[
\begin{array}{cc}
\lambda ^{-1}\overline{\lambda }I & -\overline{\mathbf{\Lambda }}^{\mathrm{T}%
}\overline{\mathbf{\Phi }}%
\end{array}%
\right] ,\text{ }S_{2}=\left[
\begin{array}{cc}
-\lambda ^{-1}\overline{\lambda }L & -\overline{\mathbf{\Lambda }}^{\mathrm{T%
}}\Psi%
\end{array}%
\right] \\
\mathbf{w}_{t-N-1}^{t-1}=[w_{t-N-1},\mathbf{w}_{t-N}^{t-1}],\text{ }\mathbf{v%
}_{t-N-1}^{t}=[v_{t-N-1},\mathbf{v}_{t-N}^{t}].%
\end{array}%
\end{equation*}
\end{proposition}

\begin{proof}
The proof follows a similar procedure with that of Theorem 1 of \cite%
{Johansen2010} and is skipped here.
\end{proof}

\begin{theorem}
\label{th1}Assume $A_{L}$ is Schur stable, $\lambda \in (0,1)$, $\overline{%
\mu }$ and $\mu $ take different positive values. Then we have (i) $%
\overline{A}_{L}$ is Schur stable; (ii) without process and measurement
noises, the estimation error exponentially converges to zero; (iii) without
process and measurement noises, when either of the following two conditions
is satisfied, $0<\mu <\overline{\mu },$ $\overline{\mu }<\mu \leq \overline{%
\mu }+\frac{1}{\lambda },$ the decaying rate of the estimation error is
monotonically increasing w.r.t. $\lambda $.
\end{theorem}

\begin{proof}
(i): If $A_{L}$ is Schur stable, for $Q_{L}>0,$ there exists a unique
solution $P_{L}>0$ satisfying the following Lyapunov equation%
\begin{equation}
A_{L}^{\mathrm{T}}P_{L}A_{L}-P_{L}+Q_{L}=0.  \label{lyapunov}
\end{equation}%
From Proposition \ref{propoopyopy}, one has%
\begin{equation*}
\overline{A}_{L}=\lambda ^{-1}\overline{\lambda }\overline{S}^{-1}A_{L}=(I+%
\widetilde{\lambda }\overline{\mathbf{\Lambda }}^{\mathrm{T}}\overline{%
\mathbf{\Lambda }})^{-1}A_{L},
\end{equation*}%
with%
\begin{equation}
\widetilde{\lambda }=\frac{\lambda }{\overline{\lambda }}.
\label{lamdatilde}
\end{equation}%
With the above, (\ref{lyapunov}) can be rewritten as%
\begin{equation*}
\overline{A}_{L}^{\mathrm{T}}(I+\widetilde{\lambda }\overline{\mathbf{%
\Lambda }}^{\mathrm{T}}\overline{\mathbf{\Lambda }})P_{L}(I+\widetilde{%
\lambda }\overline{\mathbf{\Lambda }}^{\mathrm{T}}\overline{\mathbf{\Lambda }%
})\overline{A}_{L}-\overline{P}_{L}+Q_{L}=0,
\end{equation*}%
i.e., $\overline{A}_{L}^{\mathrm{T}}\overline{P}_{L}\overline{A}_{L}-%
\overline{P}_{L}+\overline{Q}_{L}=0,$ with%
\begin{equation*}
\begin{array}{l}
\overline{P}_{L}=(I+\widetilde{\lambda }\overline{\mathbf{\Lambda }}^{%
\mathrm{T}}\overline{\mathbf{\Lambda }})P_{L}(I+\widetilde{\lambda }%
\overline{\mathbf{\Lambda }}^{\mathrm{T}}\overline{\mathbf{\Lambda }})>0, \\
\overline{Q}_{L}=(I+\widetilde{\lambda }\overline{\mathbf{\Lambda }}^{%
\mathrm{T}}\overline{\mathbf{\Lambda }})P_{L}(I+\widetilde{\lambda }%
\overline{\mathbf{\Lambda }}^{\mathrm{T}}\overline{\mathbf{\Lambda }}%
)-P_{L}+Q_{L}.%
\end{array}%
\end{equation*}%
Thus, part (i) can be established if $\overline{Q}_{L}>0.$ Given $Q_{L}>0$,
the positive definiteness of $\overline{Q}_{L}$ depends on $(I+\widetilde{%
\lambda }\overline{\mathbf{\Lambda }}^{\mathrm{T}}\overline{\mathbf{\Lambda }%
})P_{L}(I+\widetilde{\lambda }\overline{\mathbf{\Lambda }}^{\mathrm{T}}%
\overline{\mathbf{\Lambda }})-P_{L}.$ Since $(A,C)$ is observable and $N\geq
n$, one has $\overline{\mathbf{\Lambda }}^{\mathrm{T}}\overline{\mathbf{%
\Lambda }}>0.$ Therefore, we have%
\begin{equation*}
\overline{Q}_{L}=\widetilde{\lambda }(P_{L}\overline{\mathbf{\Lambda }}^{%
\mathrm{T}}\overline{\mathbf{\Lambda }}+\overline{\mathbf{\Lambda }}^{%
\mathrm{T}}\overline{\mathbf{\Lambda }}P_{L})+\widetilde{\lambda }^{2}%
\overline{\mathbf{\Lambda }}^{\mathrm{T}}\overline{\mathbf{\Lambda }}P_{L}%
\overline{\mathbf{\Lambda }}^{\mathrm{T}}\overline{\mathbf{\Lambda }}+Q_{L}.
\end{equation*}%
Note that $\widetilde{\lambda }^{2}\overline{\mathbf{\Lambda }}^{\mathrm{T}}%
\overline{\mathbf{\Lambda }}P_{L}\overline{\mathbf{\Lambda }}^{\mathrm{T}}%
\overline{\mathbf{\Lambda }}>0.$ We will just have to prove that $P_{L}%
\overline{\mathbf{\Lambda }}^{\mathrm{T}}\overline{\mathbf{\Lambda }}+%
\overline{\mathbf{\Lambda }}^{\mathrm{T}}\overline{\mathbf{\Lambda }}%
P_{L}>0. $ To do so, we show all the eigenvalues of $P_{L}\overline{\mathbf{%
\Lambda }}^{\mathrm{T}}\overline{\mathbf{\Lambda }}$ are positive. Assume $%
\gamma $ is an eigenvalue of $P_{L}\overline{\mathbf{\Lambda }}^{\mathrm{T}}%
\overline{\mathbf{\Lambda }}$ with the associated eigenvector $\tau \neq 0,$
i.e., $P_{L}\overline{\mathbf{\Lambda }}^{\mathrm{T}}\overline{\mathbf{%
\Lambda }}\tau =\gamma \tau $. Multiplying both sides of the above equation
from the left by $\overline{\mathbf{\Lambda }}^{\mathrm{T}}\overline{\mathbf{%
\Lambda }}$ gives us $\overline{\mathbf{\Lambda }}^{\mathrm{T}}\overline{%
\mathbf{\Lambda }}P_{L}\overline{\mathbf{\Lambda }}^{\mathrm{T}}\overline{%
\mathbf{\Lambda }}\tau =\gamma \overline{\mathbf{\Lambda }}^{\mathrm{T}}%
\overline{\mathbf{\Lambda }}\tau $. This further implies that $\tau ^{%
\mathrm{T}}\overline{\mathbf{\Lambda }}^{\mathrm{T}}\overline{\mathbf{%
\Lambda }}P_{L}\overline{\mathbf{\Lambda }}^{\mathrm{T}}\overline{\mathbf{%
\Lambda }}\tau =\gamma \tau ^{\mathrm{T}}\overline{\mathbf{\Lambda }}^{%
\mathrm{T}}\overline{\mathbf{\Lambda }}\tau $, i.e.,%
\begin{equation*}
\gamma =\frac{\tau ^{\mathrm{T}}\overline{\mathbf{\Lambda }}^{\mathrm{T}}%
\overline{\mathbf{\Lambda }}P_{L}\overline{\mathbf{\Lambda }}^{\mathrm{T}}%
\overline{\mathbf{\Lambda }}\tau }{\tau ^{\mathrm{T}}\overline{\mathbf{%
\Lambda }}^{\mathrm{T}}\overline{\mathbf{\Lambda }}\tau }>0,
\end{equation*}%
given $\overline{\mathbf{\Lambda }}^{\mathrm{T}}\overline{\mathbf{\Lambda }}%
P_{L}\overline{\mathbf{\Lambda }}^{\mathrm{T}}\overline{\mathbf{\Lambda }}>0$
and $\overline{\mathbf{\Lambda }}^{\mathrm{T}}\overline{\mathbf{\Lambda }}>0$%
. Thus, all the eigenvalues of $P_{L}\overline{\mathbf{\Lambda }}^{\mathrm{T}%
}\overline{\mathbf{\Lambda }}$ are positive, although $P_{L}\overline{%
\mathbf{\Lambda }}^{\mathrm{T}}\overline{\mathbf{\Lambda }}$ is not
symmetric in general. Similarly, all the eigenvalues of $\overline{\mathbf{%
\Lambda }}^{\mathrm{T}}\overline{\mathbf{\Lambda }}P_{L}$ are positive,
i.e., $\widetilde{\lambda }(P_{L}\overline{\mathbf{\Lambda }}^{\mathrm{T}}%
\overline{\mathbf{\Lambda }}+\overline{\mathbf{\Lambda }}^{\mathrm{T}}%
\overline{\mathbf{\Lambda }}P_{L})>0$. Thus, one has $\overline{Q}_{L}>0$
and $\overline{A}_{L}$ is Schur stable. (ii) Without process and measurement
noises, the estimation error dynamics reduces to%
\begin{equation}
e_{t+1}=\overline{A}_{L}e_{t}  \label{estimadyc}
\end{equation}%
for $t=0,1,2,3,\cdots $. Based on part (i), the estimation error
exponentially converges to zero. (iii) Denote%
\begin{equation*}
V_{e_{t}}=e_{t}^{\mathrm{T}}\overline{P}_{L}e_{t}
\end{equation*}%
as the Lyapunov function for (\ref{estimadyc}). It holds that%
\begin{equation*}
\Delta V_{e_{t}}=V_{e_{t+1}}-V_{e_{t}}=-e_{t}^{\mathrm{T}}\overline{Q}%
_{L}e_{t}.
\end{equation*}%
Therefore, part (iii) will be proved if $\frac{d\overline{Q}_{L}}{d\lambda }%
>0$. For $\overline{Q}_{L},$ we have%
\begin{equation*}
\frac{d\overline{Q}_{L}}{d\lambda }=(P_{L}\overline{\mathbf{\Lambda }}^{%
\mathrm{T}}\overline{\mathbf{\Lambda }}+\overline{\mathbf{\Lambda }}^{%
\mathrm{T}}\overline{\mathbf{\Lambda }}P_{L})\frac{d\widetilde{\lambda }}{%
d\lambda }+\overline{\mathbf{\Lambda }}^{\mathrm{T}}\overline{\mathbf{%
\Lambda }}P_{L}\overline{\mathbf{\Lambda }}^{\mathrm{T}}\overline{\mathbf{%
\Lambda }}\frac{d\widetilde{\lambda }^{2}}{d\lambda }.
\end{equation*}%
It can be verified that%
\begin{equation}
\frac{d\widetilde{\lambda }}{d\lambda }=\frac{\overline{\mu }}{\overline{%
\lambda }^{2}},\text{ }\frac{d(\widetilde{\lambda }^{2})}{d\lambda }=\frac{%
2\lambda (1-\lambda (\mu -\overline{\mu }))}{\overline{\lambda }^{3}}.
\label{derivimportant}
\end{equation}%
Obviously, one has $\frac{d\widetilde{\lambda }}{d\lambda }>0$. Therefore,
if $1-\lambda (\mu -\overline{\mu })\geq 0,$ one always has $\frac{d%
\overline{Q}_{L}}{d\lambda }>0.$ Given $\lambda \in (0,1)$ and $\mu \neq
\overline{\mu },$ there are two conditions under which the above inequality
holds, namely, $0<\mu <\overline{\mu },$ or $\overline{\mu }<\mu \leq
\overline{\mu }+\frac{1}{\lambda }$.
\end{proof}

Given $\overline{A}_{L}=\lambda ^{-1}\overline{\lambda }\overline{S}%
^{-1}A_{L}$ and $\overline{S}>0,$ even it has been proved in Proposition \ref%
{prop1} that $\frac{d(\lambda ^{-1}\overline{\lambda })}{d\lambda }<$ $0$,
it is not necessarily true that the magnitude of the eigenvalues of $%
\overline{A}_{L}$ will decrease with one increases $\lambda $. This can be
more clearly seen as follows. Note that we can rewrite%
\begin{equation*}
\overline{A}_{L}=(1+(\lambda ^{-1}\overline{\lambda })^{-1}\overline{\mathbf{%
\Lambda }}^{\mathrm{T}}\overline{\mathbf{\Lambda }})^{-1}A_{L}=\widetilde{%
\mathbf{\Lambda }}^{-1}A_{L}.
\end{equation*}%
Therefore, it holds that%
\begin{equation*}
\begin{array}{l}
\frac{d\overline{A}_{L}}{d\lambda }=-\widetilde{\mathbf{\Lambda }}^{-1}\frac{%
d\widetilde{\mathbf{\Lambda }}}{d\lambda }\widetilde{\mathbf{\Lambda }}%
^{-1}A_{L} \\
=-\widetilde{\mathbf{\Lambda }}^{-1}\frac{d(\lambda ^{-1}\overline{\lambda }%
)^{-1}}{d\lambda }\overline{\mathbf{\Lambda }}^{\mathrm{T}}\overline{\mathbf{%
\Lambda }}\widetilde{\mathbf{\Lambda }}^{-1}A_{L} \\
=\overrightarrow{\mathbf{\Lambda }}A_{L},%
\end{array}%
\end{equation*}%
where, $\overrightarrow{\mathbf{\Lambda }}<0$, since $\frac{d(\lambda ^{-1}%
\overline{\lambda })}{d\lambda }<$ $0$. However, whether the magnitudes of
the eigenvalues of $\overline{A}_{L}$ increase or decrease can not be
verified. What we can prove is that%
\begin{equation*}
\frac{d\Delta V_{e_{t}}}{d\lambda }=-e_{t}^{\mathrm{T}}\frac{d\overline{Q}%
_{L}}{d\lambda }e_{t}<0,
\end{equation*}%
since one always has $\frac{d\overline{Q}_{L}}{d\lambda }>0\ $, e.g., the
decaying rate of the error dynamics without disturbances can be increased
monotonically by increasing $\lambda $ with proper selections of $\mu $ and $%
\overline{\mu }$.

In the following, we consider the case with bounded disturbances. Denote%
\begin{equation*}
z_{w}=\max\limits_{w_{t}\in \mathcal{W}}\left\Vert w_{t}\right\Vert ,\text{ }%
z_{v}=\max\limits_{v_{t}\in \mathcal{V}}\left\Vert v_{t}\right\Vert .
\end{equation*}%
Given $(A,C)$ is observable in $N$ steps, one has $\overline{\mathbf{\Lambda
}}^{\mathrm{T}}\overline{\mathbf{\Lambda }}>0$ with $\overline{\mathbf{%
\Lambda }}$ being defined in (\ref{ytthat}). Although the eigenvalues of $%
\overline{\mathbf{\Lambda }}^{\mathrm{T}}\overline{\mathbf{\Lambda }}$ is
dependent on the choice of $L$, without loss of generality, it is reasonable
for us to assume that%
\begin{equation}
\overline{\mathbf{\Lambda }}^{\mathrm{T}}\overline{\mathbf{\Lambda }}\geq
\eta I  \label{assum}
\end{equation}%
where $\eta $ is a certain positive number.

\begin{theorem}
\label{th1copy}Assume $A_{L}$ is Schur stable, $\lambda \in (0,1)$, $%
\overline{\mu }$ and $\mu $ take different positive values. We have:

(i) for $t=N,N+1$,$\cdots ,$ the estimation error is bounded by $\left\Vert
e_{t-N}\right\Vert \leq \zeta _{t-N},$ where, $\zeta _{t-N}$ is a sequence
generated by%
\begin{equation}
\zeta _{t}=a\zeta _{t-1}+b\text{, }\zeta _{0}=b_{0},  \label{sequence}
\end{equation}%
in which,%
\begin{equation*}
\begin{array}{l}
a=a_{l}\left\Vert \lambda ^{-1}\overline{\lambda }\overline{S}%
^{-1}\right\Vert ,\text{ }b=\left\Vert \lambda ^{-1}\overline{\lambda }%
\overline{S}^{-1}\right\Vert \overline{z}+\left\Vert \overline{S}%
^{-1}\right\Vert \overline{\theta }, \\
b_{0}=\left\Vert I-\overline{S}^{-1}\overline{\mathbf{\Lambda }}^{\mathrm{T}}%
\overline{\mathbf{\Lambda }}\right\Vert \left\Vert x_{0}\right\Vert
+\left\Vert \lambda ^{-1}\overline{\lambda }\overline{S}^{-1}\right\Vert
\left\Vert \overline{x}_{0}\right\Vert +\left\Vert \overline{S}%
^{-1}\right\Vert \overline{\theta },%
\end{array}%
\end{equation*}%
with%
\begin{equation*}
\begin{array}{l}
\overline{z}=z_{w}+lz_{v},\text{ }\overline{\theta }=\theta _{1}\sqrt{N}%
z_{w}+\theta _{2}\sqrt{N+1}z_{v}, \\
a_{l}=\left\Vert A_{L}\right\Vert ,\text{ }\theta _{1}=\left\Vert \overline{%
\mathbf{\Lambda }}^{\mathrm{T}}\overline{\mathbf{\Phi }}\right\Vert ,\text{ }%
l=\left\Vert L\right\Vert ,\text{ }\theta _{2}=\left\Vert \overline{\mathbf{%
\Lambda }}^{\mathrm{T}}\Psi \right\Vert .%
\end{array}%
\end{equation*}

(ii) if%
\begin{equation}
\frac{a_{l}\sqrt{n}}{1+\widetilde{\lambda }\eta }<1,  \label{conditiona}
\end{equation}%
with $n$ as the dimension of the state, and $\widetilde{\lambda },\eta
,a_{l} $ being defined in (\ref{lamdatilde}), (\ref{assum}), (\ref{sequence}%
), respectively, then $a<1\ $and the sequence $\left\{ \zeta _{t}\right\} $
converges exponentially to%
\begin{equation*}
\zeta _{\infty }=\frac{b}{1-a}.
\end{equation*}
\end{theorem}

\begin{proof}
(i) This part can be proved by using the error dynamics given in Proposition %
\ref{propoopyopy}, following similar steps as those in Theorem 3 of \cite%
{Johansen2010}, and recognizing that $\left\Vert A+B\right\Vert \leq
\left\Vert A\right\Vert +\left\Vert B\right\Vert $, $\left\Vert
AB\right\Vert \leq \left\Vert A\right\Vert \left\Vert B\right\Vert ,$ where,
$A$ and $B$ are vectors or matrices of proper dimensions. (ii) The
convergence of sequence $\left\{ \zeta _{t}\right\} $ (\ref{sequence})
depends on $a<1.$ When $a<1$, $\zeta _{\infty }=\frac{b}{1-a},$ as $%
t\rightarrow \infty .$ Note that%
\begin{equation*}
\begin{array}{l}
a=a_{l}\left\Vert \lambda ^{-1}\overline{\lambda }S^{-1}\right\Vert =a_{l}%
\sqrt{tr(\lambda ^{-1}\overline{\lambda }S^{-1}\lambda ^{-1}\overline{%
\lambda }S^{-1})} \\
=a_{l}\sqrt{tr((I+\widetilde{\lambda }\overline{\mathbf{\Lambda }}^{\mathrm{T%
}}\overline{\mathbf{\Lambda }})^{-1}(I+\widetilde{\lambda }\overline{\mathbf{%
\Lambda }}^{\mathrm{T}}\overline{\mathbf{\Lambda }})^{-1})},%
\end{array}%
\end{equation*}%
with $\widetilde{\lambda }$ being defined in (\ref{lamdatilde}). Note that%
\begin{equation*}
\begin{array}{l}
\lambda _{\mathrm{\min }}(M_{1})\lambda _{\mathrm{\max }}(M_{2})\leq \lambda
_{\mathrm{\min }}(M_{1})tr(M_{2}) \\
\leq tr(M_{1}M_{2})\leq \lambda _{\mathrm{\max }}(M_{1})tr(M_{2}),%
\end{array}%
\end{equation*}%
for two positive definite matrices $M_{1}$ and $M_{2}$ \cite{Athans}. Since $%
\overline{\mathbf{\Lambda }}^{\mathrm{T}}\overline{\mathbf{\Lambda }}\geq
\eta I,$ one has that%
\begin{equation*}
I+\widetilde{\lambda }\overline{\mathbf{\Lambda }}^{\mathrm{T}}\overline{%
\mathbf{\Lambda }}\geq (1+\widetilde{\lambda }\eta )I,
\end{equation*}%
i.e.,%
\begin{equation*}
tr((I+\widetilde{\lambda }\overline{\mathbf{\Lambda }}^{\mathrm{T}}\overline{%
\mathbf{\Lambda }})^{-1})\leq n\lambda _{\mathrm{\max }}((I+\widetilde{%
\lambda }\overline{\mathbf{\Lambda }}^{\mathrm{T}}\overline{\mathbf{\Lambda }%
})^{-1})=\frac{n}{1+\widetilde{\lambda }\eta }.
\end{equation*}%
It can be further derived that%
\begin{equation*}
a\leq a_{l}\sqrt{\lambda _{\mathrm{\max }}((I+\widetilde{\lambda }\overline{%
\mathbf{\Lambda }}^{\mathrm{T}}\overline{\mathbf{\Lambda }})^{-1})\frac{n}{1+%
\widetilde{\lambda }\eta }}\leq \frac{a_{l}\sqrt{n}}{1+\widetilde{\lambda }%
\eta }.
\end{equation*}%
Therefore, if (\ref{conditiona}) is satisfied, we will have $a<1.$
\end{proof}

An important question here is whether $\lambda $ can be tuned to possibly
reduce $\zeta _{\infty }.$ This motivates us to derive%
\begin{equation*}
\frac{d\zeta _{\infty }}{d\lambda }=\frac{1}{(1-a)^{2}}\left[ (1-a)\frac{db}{%
d\lambda }+b\frac{da}{d\lambda }\right] .
\end{equation*}%
It is not obvious to ascertain whether $\frac{d\zeta _{\infty }}{d\lambda }$
is positive or negative when $\lambda \in (0,1)$, because $a$ and $b$ are
parameterized by $\lambda $ in a complicated way. To circumvent this
difficulty, we set our sights lower and define a normalized sequence by
dividing both sides of (\ref{sequence}) with $b$: $\overline{\zeta }_{t}=a%
\overline{\zeta }_{t-1}+1$, $\overline{\zeta }_{0}=\frac{b_{0}}{b},$ in
which,%
\begin{equation*}
\overline{\zeta }_{t}=\frac{\zeta _{t}}{b},\overline{\zeta }_{t-1}=\frac{%
\zeta _{t-1}}{b}.
\end{equation*}%
If we define $\overline{e}_{t-N}=\frac{e_{t-N}}{b},$ it holds that $%
\left\Vert \overline{e}_{t-N}\right\Vert \leq \overline{\zeta }_{t-N}.$ We
then have the following results.

\begin{theorem}
\label{th1co}Assume $A_{L}$ is Schur stable, $\lambda \in (0,1)$, $\overline{%
\mu }$ and $\mu $ take different positive values. If (\ref{conditiona}) is
satisfied, the normalized bounding sequence $\left\{ \overline{\zeta }%
_{t}\right\} $ converges exponentially to the following value $\overline{%
\zeta }_{\infty }=\frac{1}{1-a}$. Moreover, $\overline{\zeta }_{\infty }$
decreases monotonically when one increases $\lambda .$
\end{theorem}

\begin{proof}
The convergence of $\left\{ \overline{\zeta }_{t}\right\} $ to $\overline{%
\zeta }_{\infty }$ can be proved similarly to that of Theorem \ref{th1copy}.
Note that for a given $M>0$, if $\gamma >0$ is an eigenvalue of $M$ with the
associated eigenvector $\tau \neq 0,$ i.e., $M\tau =\gamma \tau ,$ then $%
\gamma ^{2}$ is an eigenvalue of $M^{2}$ with the associated eigenvector $%
\tau \neq 0$, and $\gamma ^{-1}$ is an eigenvalue of $M^{-1}$ with the
associated eigenvector $\tau \neq 0$. These can be easily verified as follows%
\begin{equation*}
M^{2}\tau =\gamma M\tau =\gamma ^{2}\tau ,\text{ }M^{-1}\tau =\gamma
^{-1}\tau .
\end{equation*}%
Denote $\left\{ \gamma _{1},\gamma _{2},\cdots ,\gamma _{n-1},\gamma
_{n}\right\} $ as the eigenvalues of the matrix $\overline{\lambda }%
S^{-1}=(I+\widetilde{\lambda }\overline{\mathbf{\Lambda }}^{\mathrm{T}}%
\overline{\mathbf{\Lambda }})^{-1}.$ Then we know that%
\begin{equation*}
tr((I+\widetilde{\lambda }\overline{\mathbf{\Lambda }}^{\mathrm{T}}\overline{%
\mathbf{\Lambda }})^{-2})=\gamma _{1}^{2}+\gamma _{2}^{2}+\cdots +\gamma
_{n-1}^{2}+\gamma _{n}^{2}.
\end{equation*}%
Note from (\ref{derivimportant}), we have%
\begin{equation*}
\frac{d\widetilde{\lambda }}{d\lambda }=\frac{\overline{\mu }}{\overline{%
\lambda }^{2}}>0.
\end{equation*}%
Therefore, eigenvalues of matrix $I+\widetilde{\lambda }\overline{\mathbf{%
\Lambda }}^{\mathrm{T}}\overline{\mathbf{\Lambda }}$ increase monotonically
when one increases $\lambda .$ In other words, eigenvalues of matrix $(I+%
\widetilde{\lambda }\overline{\mathbf{\Lambda }}^{\mathrm{T}}\overline{%
\mathbf{\Lambda }})^{-1}$ and $(I+\widetilde{\lambda }\overline{\mathbf{%
\Lambda }}^{\mathrm{T}}\overline{\mathbf{\Lambda }})^{-2}$ decrease
monotonically when $\lambda $ is increased$,$ i.e,%
\begin{equation*}
\begin{array}{l}
\frac{da}{d\lambda }=a_{l}\frac{d\left\Vert \lambda ^{-1}\overline{\lambda }%
S^{-1}\right\Vert }{d\lambda } \\
=a_{l}\frac{d\left\Vert (I+\widetilde{\lambda }\overline{\mathbf{\Lambda }}^{%
\mathrm{T}}\overline{\mathbf{\Lambda }})^{-1}\right\Vert }{d\lambda } \\
=a_{l}\frac{d\sqrt{tr((I+\widetilde{\lambda }\overline{\mathbf{\Lambda }}^{%
\mathrm{T}}\overline{\mathbf{\Lambda }})^{-2})}}{d\lambda }<0.%
\end{array}%
\end{equation*}%
The fact that $\overline{\zeta }_{\infty }$ is a monotonically decreasing
function of $\lambda $ can be obtained by recognizing%
\begin{equation*}
\frac{d\overline{\zeta }_{\infty }}{d\lambda }=\frac{1}{(1-a)^{2}}\frac{da}{%
d\lambda }<0.
\end{equation*}
\end{proof}

\section{\label{exam}Illustrative Example}

We illustrate the results with the land-based vehicle example in \cite%
{Bitmead2007}, where detailed information about the system model can be
found. In the above model, the first two components of the states are the
northerly and easterly positions and the last two are the northerly and
easterly velocities, respectively. Different from \cite{Bitmead2007}, the
heading of the vehicle is assumed to be unknown. Assume the input is zero
and the sampling period $T=0.5$ s. We design a pre-estimator with%
\begin{equation*}
L=\left[
\begin{array}{ccc}
1.2466 & 0 & 0 \\
0 & 0.8627 & 0.4358 \\
0.6759 & 0 & 0 \\
0 & 0.0090 & 0.8535%
\end{array}%
\right]
\end{equation*}%
so that the set of eigenvalues of $A_{L}=A-LC$ is $\left\{
0.1519,0.6015,0.1419\pm 0.0236i\right\} $.

Firstly, for system (\ref{augmented}), we compare the metamorphic MHE in
Section \ref{copy} (with $\lambda $ taking the value of $0.1$, $0.5$) with
the the unbiased FIR filter in \cite{Zhao2017}. As such, we choose $Q$ $%
=I_{4}$; $R=0.5I_{3}$; $M=10I_{7}$, and the rolling horizon length to be $20$
in both metamorphic MHE and FIR. For metamorphic MHE, the weightings on the
initial state estimation error within the rolling horizon is taken to be the
steady-state solution to the ARE (\ref{AREaugment}), with $\lambda $ taking
different values of $0.1$ and $0.5$, respectively. Assume that elements of $%
w_{k}$ and $\nu _{k}$ are uniformly distributed numbers between $\left[
-0.1,0.1\right] $ and $\left[ -0.25,0.25\right] $, respectively. The initial
state guess of the system (\ref{augmented}) is a realization of zero mean,
normally distributed random variable with unit covariance while the true
value is $[5\cdot \mathbf{1}_{4},10\cdot \mathbf{1}_{4}]$. The evolutions of
the northerly position estimation error for both MHE and FIR\ are
illustrated in Figure \ref{final1}. It can be seen that (i) the metamorphic
MHE strategy with $\lambda =0.5$ renders better performance than the
strategy with $\lambda =0.1$ in the first few sampling instants; afterwards,
their performance are nearly distinguishable (because of the effect of the
pre-estimator); (ii) the performance of FIR is roughly comparable with that
of MHE, and the FIR seems to give slightly worse performance than MHE after
the $25$-th sampling instant. In Figure \ref{final1}, before the $21$-st
sampling instant, both metamorphic MHE and FIR work as batch least-squares
estimators with MHE using and FIR not using the initial guess, respectively.

\begin{figure}[tbph]
\centering
\epsfig{figure= 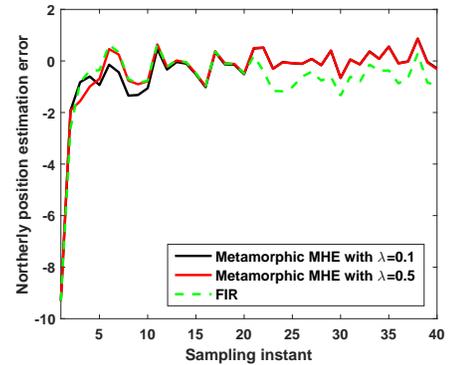,scale=.45, bb=110 182 500 535}
\caption{The evolutions of estimation error}
\label{final1}
\end{figure}






Secondly, using the same example, we compare the MHE strategy in Section \ref%
{main2} with the the unbiased FIR filter in \cite{Zhao2017}. For the
trajectories of the estimation errors, we notice similar patterns with those
shown in Figure \ref{final1}. Thus, further illustrations are not presented
here. Instead, we conduct some numerical analysis of the estimation errors
for both strategies. As such, we select $\mu =0.15$ and $\overline{\mu }=0.1$
so that the condition in Theorem \ref{th1} is maintained for all $\lambda
\in (0,1).$ For MHE, assume that both $x_{0}$ (the true state) and $%
\overline{x}_{0}$ (the initial guess) are a realization of zero mean,
normally distributed random variable with unit covariance. We consider two
cases, where elements of the noise vectors are bounded by: (1) $\left[
-0.01,0.01\right] $, (2) $\left[ -0.025,0.025\right] $, respectively. We
take the initial step as the 0-th time instant. Starting from the $20$-th
time instant, we calculate the estimation error for $100$ steps when $%
\lambda $ takes the value of $0,0.25$, $0.5$, and $0.75$, respectively. When
$\mu =0.15$ and $\overline{\mu }=0.1,$ it can be verified that from
Proposition \ref{prop1} that $(1-\lambda )\overline{\mu }<\lambda (1-\mu )$
for $0.25$, $0.5$, $0.75,$ therefore, the MHE has a forgetting towards the
pre-estimate. We repeat the above simulation for 1000 random scenarios. The
average root mean square errors (ARMSE) for the MHE and FIR strategies, are
shown in Table \ref{table1}. From Table \ref{table1}, we can conclude: (i)
with the increase of $\lambda $, MHE gradually improves performance; (ii)
the unbiased FIR filter performs roughly as well as MHE with $\lambda \in
\left( 0.25,0.5\right) $, although theoretical comparisons between the two
are hard to obtain.
\begin{table}[t]
\caption{ARMSE comparison of different scenarios ($\protect\mu =0.15$, $%
\overline{\protect\mu }=0.1$)}
\label{table1}\renewcommand{\arraystretch}{1.3} \centering%
\begin{tabular}{cccccc}
\hline
case & $\lambda =0$ & $\lambda =0.25$ & $\lambda =0.5$ & $\lambda =0.75$ &
FIR \\
1 & 0.388 & 0.162 & 0.113 & 0.091 & 0.121 \\
2 & 0.3901 & 0.163 & 0.116 & 0.095 & 0.151%
\end{tabular}%
\end{table}

\section{\label{conclusion}Conclusion}

We have proposed a MHE methodology with pre-estimation and normalized
forgetting/discounting effects. This is achieved by the introduction of a
cost formulation parameterized by a design parameter $\lambda \in \lbrack
0,1]$. We have examined the idea in two general MHE frameworks. When $%
\lambda =0,$ the proposed technique reduces to the existing estimator. When $%
\lambda $ is increased, the technique has a more aggressive forgetting
effect towards the old data, in a normalized sense. Therefore, when one
increases $\lambda $, the proposed framework gradually improves the
estimation performance, based on the pre-estimator. Extension of the method
to the nonlinear case poses no conceptual difficulty, although establishment
of theoretical results would be more involved.

\section{Acknowledgment}

The authors thank the reviewers and Editors for their constructive comments
which helped to improve this paper's quality. The first author is grateful
to Prof. Graham Goodwin and A/Prof Maria Seron at University of Newcastle,
Australia, for sharing illuminating thoughts on constrained estimation. He
would also like to thank Dr. Shunyi Zhao at Jiangnan University, China, for
helpful discussions on finite impulse response filtering.


\begin{thebibliography}{99}
\bibitem{Sukkarieh2001} G. Dissanayake, S. Sukkarieh, E. Nebot, and H.
Durrant-Whyte, The aiding of a low-cost strapdown inertial measurement unit
using vehicle model constraints for land vehicle applications, \textit{IEEE
Trans. on Robotics and Automation}, Vol. 17, No. 5, pp. 731-747, 2001.

\bibitem{Mahata2004} K. Mahata and T. S\"{o}derstr\"{o}m, Improved
estimation performance using known linear constraints, \textit{Automatica},
Vol. 40, No. 8, pp. 1307--1318, 2004.

\bibitem{Bitmead2007} S. Ko and R. R. Bitmead, State estimation for linear
systems with state equality constraints, \textit{Automatica}, Vol. 43, No.
8, pp. 1363--1368, 2007.

\bibitem{Chen2007} D. Chu , T. Chen, and H. J. Marquez, Robust moving
horizon state observer, \textit{IJC, }Vol. 80, No. 10, pp. 1636--1650, 2007.

\bibitem{Mayne2009} J. B. Rawlings and D. Q. Mayne, Model predictive
control: theory and design, LLC, Madison: Nob Hill Publishing, 2009.

\bibitem{Boyd2012} E. Chu, A. Keshavarz, D. Gorinevsky, and S. Boyd, Moving
horizon estimation for staged QP problems, \textit{Proc. of IEEE CDC}, pp.
3177-3182, Hawaii, 2012.

\bibitem{Kerrigan2016} M. Ge and E. C. Kerrigan, Relations between full
information and Kalman-based estimation, \textit{Proc. of IEEE CDC}, pp.
2041-2046, Las Vegas, 2016.

\bibitem{Rawlings2001} C. V. Rao, J. B. Rawlings, and J. H. Lee, Constrained
linear state estimation--a moving horizon approach, \textit{Automatica},
Vol. 37, No. 10, pp. 1619-1628, 2001.

\bibitem{Kong2018} H. Kong and S. Sukkarieh, Suboptimal receding horizon
estimation via noise blocking, \textit{Automatica, }Provisionally accepted,
2018.

\bibitem{Alessandri2003} A. Alessandri, M. Baglietto, and G. Battistelli,
Receding horizon estimation for discrete time linear systems, \textit{IEEE
Trans. Autom. Control}, Vol. 48, No. 3, pp. 473--478, 2003.

\bibitem{Johansen2010} D. Sui, T. A. Johansen, and L. Feng, Linear moving
horizon estimation with pre-estimating observer, \textit{IEEE Trans. Autom.
Control}, Vol. 55, No. 10, pp. 2363--2368, 2010.

\bibitem{Johansen2014} D. Sui and T. A. Johansen, Linear constrained moving
horizon estimator with pre-estimating observer, \textit{Systems \& Control
Letters}, Vol. 67, pp. 40--45, 2014.

\bibitem{Zhao2017} Y. S. Shmaliy, S. Zhao, and C. K. Ahn, Unbiased finite
impluse response filtering: an iterative alternative to Kalman filtering
ignoring noise and initial conditions, \textit{IEEE Control Systems Magazine}%
, Vol. 37, No. 5, pp. 70-89, 2017.

\bibitem{Shi2016} C. K. Ahn, P. Shi, and M. V. Basin, Deadbeat dissipative
FIR filtering, \textit{IEEE Trans. on Circuits and Systems I: Regular Papers}%
, Vol. 63, No. 8, pp. 1210-1221, 2016.

\bibitem{Kong2012} H. Kong, G. C. Goodwin, and M. M. Seron, A revisit to
inverse optimality of linear systems, \textit{IJC, }Vol. 85, No. 10, pp.
1506--1514, 2012.

\bibitem{Kong2013} H. Kong, G. C. Goodwin, and M. M. Seron, Predictive
metamorphic control, \textit{Automatica, }Vol. 49, No. 12, pp. 3670--3676,
2013.

\bibitem{Maciejowski2013} E. N. Hartley and J. M. Maciejowski, Designing
output-feedback predictive controllers by reverse engineering existing LTI
controllers, \textit{IEEE Trans. on Automatic Control}, Vol. 58, No. 11, pp.
2934--2939, 2013.

\bibitem{Laub2005} A. J. Laub, Matrix analysis for scientists and engineers,
SIAM, 2005.

\bibitem{Bitmead1985} R. R. Bitmead, M. R. Gevers, I. R. Petersen, and R. J.
Kaye, Monotonicity and stabilizability properties of solutions of the
Riccati difference equation, \textit{System \& Control Letters}, Vol. 5, No.
5, pp. 309-315, 1985.

\bibitem{Goodwin1984} S. W. Chan, G. C. Goodwin, and K. S. Sin, Convergence
properties of the Riccati difference equation in optimal filtering of
nonstabilizable systems, \textit{IEEE Trans. Autom. Control}, Vol. 29, No.
2, pp. 110-118, 1984.

\bibitem{Athans} D. L. Kleinman and M. Athans, The design of suboptimal
linear time varying equations, \textit{IEEE Trans. Autom. Control}, Vol. 13,
No. 2, pp. 150--159, 1968.
\end{thebibliography}
\end{document}